\documentclass[a4paper]{article}

\usepackage{amsmath}
\usepackage{amsthm}
\usepackage{amssymb}
\usepackage{hyperref}
\usepackage{bbold}
\usepackage{mathrsfs}
\usepackage{algorithm}
\usepackage{algorithmic}
\usepackage{tikz}
\usepackage{tikz-cd}

\newtheorem{theorem}{Theorem}
\newtheorem{proposition}{Proposition}
\newtheorem{corollary}{Corollary}
\newtheorem{lemma}{Lemma}
\newtheorem{definition}{Definition}
\newtheorem{remark}{Remark}
\newtheorem{example}{Example}

\floatname{algorithm}{Procedure}

\newcommand{\Der}{\ensuremath{\partial}}
\newcommand{\Int}{\ensuremath{{\textstyle\int}}}
\newcommand{\PX}{\mathscr{P}(\Span{X})}
\DeclareMathOperator{\id}{id}
\DeclareMathOperator{\supp}{supp}
\DeclareMathOperator{\Hom}{Hom}
\DeclareMathOperator{\DM}{DM}

\newcommand{\R}{\ensuremath{R}}
\newcommand\K{\mathbb{K}}

\newcommand\RX{R\Span{X}}
\newcommand\KX{\K\Span{X}}
\newcommand\Span[1]{\langle #1\rangle}

\newcommand\Qord{\leq_Q}
\DeclareMathOperator{\LM}{LM}

\DeclareMathOperator{\SP}{SP}
\DeclareMathOperator{\amb}{\mathfrak{a}} 

\newcommand\rewTrans{\overset{*}{\to}}

\newcommand\ie{i.e.}
\newcommand\Gr{Gr\"obner}

\title{Compatible rewriting of noncommutative polynomials for proving operator identities}
\author{Cyrille Chenavier, Clemens Hofstadler, \and Clemens G. Raab, and Georg Regensburger\thanks{This work was supported by the Austrian Science Fund (FWF): P~27229, P~32301, and P~32952\newline \texttt{\{cyrille.chenavier,clemens.hofstadler,clemens.raab,georg.regrensburger\}@jku.at}\newline
Institute for Algebra, Johannes Kepler University, Altenberger Stra\ss e 69, Linz, Austria}}
\date{}

\begin{document}
\maketitle

\begin{abstract}
  The goal of this paper is to prove operator identities using equalities
  between noncommutative polynomials. In general, a polynomial expression
  is not valid in terms of operators, since it may not be compatible with
  domains and codomains of the corresponding operators. Recently, some of
  the authors introduced a framework based on labelled quivers to
  rigorously translate polynomial identities to operator identities. In
  the present paper, we extend and adapt the framework to the context of
  rewriting and polynomial reduction. We give a sufficient condition on
  the polynomials used for rewriting to ensure that standard polynomial
  reduction automatically respects domains and codomains of operators.
  Finally, we adapt the noncommutative Buchberger procedure to compute
  additional compatible polynomials for rewriting. In the
  package \texttt{OperatorGB}, we also provide an implementation of the
  concepts developed.
\end{abstract}


\section{Introduction}

Properties of linear operators can often be expressed in terms of identities
they satisfy. Algebraically, these identities can be represented in terms of
noncommutative polynomials in some set $X$. The elements of $X$ correspond to
basic operators and polynomial multiplication models composition
of  operators. Based on this, proving that a claimed operator identity follows
from assumed identities corresponds to the polynomial $f$, associated
to the claim, lying in the ideal generated by the set $F$ of polynomials associated to
the assumptions. However, ideal membership $f\in (F)$ is not enough for proving an
operator identity in general, since
computations with  noncommutative polynomials ignore compatibility conditions
between domains and codomains of the operators. 

In order to represent domains and codomains of
the operators, we use the framework introduced recently in~\cite{Raab2019FormalPO}.
So, we consider a quiver (\ie, a directed multigraph) $Q$, where
vertices correspond to functional spaces, edges correspond to basic operators
between those spaces and are labelled with symbols from $X$.
Then, paths in $Q$ correspond to composition of basic operators and induce monomials over $X$ that are compatible with $Q$.
Note that we can allow the same label for different edges 
if the corresponding operators satisfy the same identities in $F$. For instance,
differential and integral operators can act on different functional
spaces, as illustrated in our running example below.
For formal details and relevant notions, see Sections~\ref{sec:prelim} and \ref{sec:realizations}.
Informally, a polynomial is
compatible with the quiver if it makes sense in terms of operators and $f$ is
called a $Q$-consequence of $F$ if it can be obtained from $F$ by doing
computations using compatible polynomials only. This means that these computations
also make sense in terms of operators.

Obviously, the claim $f$ and the assumptions $F$ have to be compatible with $Q$.
In~\cite{Raab2019FormalPO}, it was shown that $f$ is a $Q$-consequence of $F$ if
$f\in (F)$ and each element of $F$ is uniformly compatible, which means that
all its monomials can be assigned the same combinations of domains and codomains.
This is in particular the case when each edge has a unique label and polynomials
do not have a constant term. Note that ideal membership can be checked
independently of $Q$ and is undecidable in general. In practice, it can often be
checked by computing a (partial) noncommutative \Gr\ basis $G$ of $F$ and reducing $f$ to zero
by $G$, see~\cite{MR1299371}.
The package \texttt{OperatorGB} \cite{hofstadler2019certifying} can check compatibitlity of polynomials with quivers and, based on partial \Gr\ bases, can compute explicit representations of polynomials in terms of generators of the ideal.
Versions for \textsc{Mathematica} and \textsc{SageMath} can be obtained at:
\begin{center}
\url{http://gregensburger.com/softw/OperatorGB}
\end{center}

In this paper, we generalize the formal definition of $Q$-con\-se\-quen\-ces to
the case when elements of $F$ are compatible but not necessarily uniformly compatible,
see Section~\ref{sec:Q-csq}. Then, we show in Section~\ref{sec:realizations} that
being a $Q$-consequence implies that the corresponding operator identity can indeed be
proven by computations with operators. Since elements of $F$ do not have to be
uniformly compatible, we impose in Section~\ref{sec:general_rew_systems}
restrictions on the polynomial rewriting, so that it respects the quiver. For the
same reason, we also impose restrictions on the computation of partial \Gr\ bases in
Section~\ref{sec:compatible_GB}.
Based on such a partial \Gr\ basis, one often can prove algorithmically that $f$ is a $Q$-consequence of $F$ just by standard polynomial reduction.
To this end, we also extend the package \texttt{OperatorGB}.

Gr\"obner bases for noncommutative polynomials have been applied to operator
identities in the pioneering work \cite{HeltonWavrik1994,HeltonStankusWavrik1998}, where \Gr\ bases are used to simplify matrix identities in linear systems theory.
In~\cite{HeltonStankus1999,Kronewitter2001}, the main strategy for solving matrix equations, coming from factorization of engineering systems and matrix completion problems, is to apply Gr\"obner bases with respect to an ordering appropriate for elimination.
The same approach was used in~\cite{RosenkranzBuchbergerEngl2003} to compute Green's operators for linear two-point boundary problems with constant coefficients.

If edges of the quiver have unique labels, it has been observed in the literature that the operations used in the noncommutative analog of Buchberger's algorithm respect compatibility of polynomials with domains and codomains of operators, cf.~\cite[Thm.~25]{HeltonStankusWavrik1998}.
See also Remark~\ref{rem:uniquelabels} and Theorem~\ref{thm:uniquelabels} for a formal statement using the framework of the present paper.
For an analogous observation in the context of path algebras, see~\cite[Sec.~47.10]{Mora2016}.
We were informed in personal communication that questions related to proving operator identities via computations of Gr\"obner bases are also addressed in \cite{LevandovskyySchmitz2019}.

Alternatively, computations with operators can also be modelled by partial algebras arising from diagrams, for which an analogous notion of Gr\"obner bases was sketched in \cite[Sec.~9]{Bergman1978} and developed in \cite{BokutChenLi2012}.
Moreover, generalizations of Gr\"obner bases and syzygies are considered in \cite{GuiraudHoffbeckMalbos2019}, where higher-dimensional linear rewriting systems are introduced for rewriting of operators with domains and codomains.

We conclude this section with a small running example that we use throughout the paper to illustrate the notions that we introduce from practical point of view.
A \textsc{Mathematica} notebook that illustrates the use of the new functionality of the package using this running example can be obtained at the webpage mentioned above.

\begin{example}
  Consider the inhomogeneous linear differential equation
   \begin{equation*}
   y^{\prime\prime}(x)+A_1(x)y^\prime(x)+A_0(x)y(x)=r(x)
   \end{equation*}
   and assume that it can be factored into the two first-order equations
   \[
   y^\prime(x)-B_2(x)y(x)=z(x) \quad\text{and}\quad z^\prime(x)-B_1(x)z(x)=r(x).
   \]
   It is well-known that a particular solution is given by the nested
   integral
   \begin{equation}\label{equ:sol}
     y(x) = H_2(x)\int_{x_2}^xH_2(t)^{-1}H_1(t)\int_{x_1}^tH_1(u)^{-1}r(u)\,du\,dt,
   \end{equation}
   where $H_i(x)$ is a solution of $y^\prime(x)-B_i(x)y(x)=0$ such that $H_i(x)^{-1}$ exists.
   In order to translate this claim into an operator identity, let us consider
   the differentiation $\partial:y(x)\mapsto y'(x)$ and the two
   integrations
   \[
    \Int_1:y(x)\mapsto\int_{x_1}^xy(t)\,dt \quad\text{and}\quad \Int_2:y(x)\mapsto\int_{x_2}^xy(t)\,dt.
   \]
   Moreover, any function $F(x)$ induces
   a multiplication operator $F:y(x)\mapsto y(x)F(x)$ and $\cdot$ denotes the composition
   of operators. Thus, the factored differential equation and the solution correspond to
   the following operators
   \[
   L:=(\Der-B_1)\cdot(\Der-B_2),\quad S:=H_2{\cdot}\Int_2{\cdot}H_2^{-1}{\cdot}H_1{\cdot}\Int_1{\cdot}H_1^{-1}
   \]
   and the claim corresponds to the identity $L\cdot S=\id$. In terms of functions,
   this means that $y(x)=(Sr)(x)$ is a solution of
   \begin{equation}\label{ex:inhom_ODE}
     (Ly)(x)=r(x).
   \end{equation}
   Using the Leibniz rule,
   $H_i$ being a solution of the factor differential equation corresponds to
   \[\Der\cdot  H_i=H_i\cdot \Der+B_i\cdot H_i\]
   and the invertibility corresponds to $H_i\cdot H_i^{-1}=\id$. The last fact we use
   for proving the claim is the fundamental theorem of calculus, which corresponds to
   \[\Der\cdot\Int_1=\id,\qquad\Der\cdot\Int_2=\id.\]
   In Example~\ref{ex:compatible}, we will show how these operator identities can be translated into noncommutative polynomials that are compatible with a quiver.
\end{example}

\section{Preliminaries}
\label{sec:prelim}

In this section, we recall the main definitions and basic facts from~\cite{Raab2019FormalPO} that formalize compatibility of polynomials with a labelled quiver.

We fix a commutative ring $R$ with unit as well as a set $X$.
We consider the free noncommutative algebra $\RX$ generated by the alphabet $X$:
it can be regarded as the
ring of noncommutative polynomials in the set of indeterminates $X$ with
coefficients in $\R$, where indeterminates commute with coefficients but
not with each other. The monomials are words
$x_1\dots x_n \in \langle{X}\rangle$, $x_i \in X$, including the empty
word $1$. Every polynomial $f \in \R\langle{X}\rangle$ has a unique
representation as a sum  
\[
 f=\sum_{m\in\langle{X}\rangle}c_mm
\]
with coefficients $c_m \in \R$, such that only finitely many coefficients
are nonzero, and its support is defined as
\[
 \supp(f):=\{m\in\langle{X}\rangle\ |\ c_m\neq0\},
\]
where $c_m$ are as above.

Recall that a quiver is a tuple $(V,E,s,t)$, where $V$ is a set of
vertices, $E$ is a set of edges, and $s,t:E\to V$ are source and target
maps, that are extend to all paths $p=e_n\cdots e_1$ by letting
$s(p)=s(e_1)$ and $t(p)=t(e_n)$. For every vertex $v\in V$, there is a
distinct path $\epsilon_v$ that starts and ends in $v$ without passing
through any edge, and which acts as a local identity on paths $p$, that
is $\epsilon_{t(p)}p=p=p\epsilon_{s(p)}$. A labelled quiver,
$Q=(V,E,X,s,t,l)$ is a quiver equipped with a label function $l:E\to X$
of edges into the alphabet $X$. We extend $l$ into a function from paths
to monomials by letting $l(p)=l(e_n)\cdots l(e_1)\in\langle{X}\rangle$,
and $l(\epsilon_v)=1$ is the empty word for every vertex $v$. From now
on, we fix a labelled quiver $Q=(V,E,X,s,t,l)$.

\begin{definition}
  Given a labelled quiver and a monomial $m$, we define the set of
  \emph{signatures} of $m$ as
 \[
 \sigma(m):=\{(s(p),t(p))\ |\ p\text{ a path in $Q$ with }l(p)=m\}
 \subseteq V\times{V}.
 \]
 A polynomial $f \in \R\langle{X}\rangle$ is said to be \emph{compatible}
 with $Q$ if its set of signatures $\sigma(f)$ is non empty, where:
 \[
  \sigma(f):=\bigcap_{m\in\supp(f)}\sigma(m) \subseteq V\times{V}.
  \]
  Finally, we denote by $s(f)$ and $t(f)$ the images of $\sigma(f)$ through the
  natural projections of $V\times V$ on $V$.
\end{definition}
Note that we have $\sigma(0)=V\times{V}$ and $\sigma(1)=\{(v,v)\ |\ v \in V\}$.

Computing with compatible polynomials does not always result in
compatible polynomials. However, under some conditions, the sum and
product of compatible polynomials are compatible as well. The following
properties of signatures are straightforward to prove; see also Lemmas
$10$ and $11$ in~\cite{Raab2019FormalPO}.

\begin{lemma}\label{lem:signatures}
  Let $f,g\in\RX$ be compatible with $Q$. Then,
  \begin{enumerate}
  \item If $\sigma(f)\cap\sigma(g)\neq\emptyset$, then $f+g$ is
    compatible with $Q$ and $\sigma(f+g)\supseteq\sigma(f)\cap\sigma(g)$.
  \item If $s(f)\cap t(g)\neq\emptyset$, then $fg$ is compatible with
    $Q$ and
    \begin{multline*}
      \hspace{0.5cm}\sigma(fg)\supseteq\{(u,w) \in s(g)\times{t(f)}\ |\ \\ \exists{v \in s(f) \cap t(g)}:
      (u,v) \in \sigma(g) \wedge (v,w) \in \sigma(f)\}.
    \end{multline*}
  \end{enumerate}
\end{lemma}

We use the following conventions when we draw labelled quivers: we do not
give names to vertices and edges, but denote them by a bullet and an arrow
oriented from its source to its target, respectively,
and the label of an edge is simply written above the arrow representing
this edge.

\begin{example}\label{ex:compatible}
  Let us continue the running example. The Leibniz rule and invertibility
  for $H_1$ and $H_2$ and the fundamental theorem of calculus correspond
  to the following noncommutative polynomials in $\mathbb{Z}\Span{X}$, where
  $X=\{h_1,h_2,b_1,b_2,\tilde{h}_1,\tilde{h}_2,i,d\}$.
  \begin{gather*}
    f_1= d h_1-h_1d-b_1h_1,\qquad f_2 = d h_2-h_2d-b_2h_2,\\
    f_3 = h_1\tilde{h}_1-1,\qquad f_4 = h_2\tilde{h}_2-1,\\
    f_5=di-1
  \end{gather*}
  We collect these polynomials in the set $F:=\{f_1,\dots,f_5\}$.
  Notice that we represent the two integrals by a single indeterminate,
  so we only need one polynomial for the fundamental theorem of
  calculus. The claim corresponds to
  \[f:=(d-b_1)(d-b_2)h_2i\tilde{h}_2h_1i\tilde{h}_1-1.\]
  Since integration and differentiation decrease and increase 
  the regularity of functions, it is natural to consider the following
  labelled quiver with $3$ vertices (more details are given
  Section~\ref{sec:realizations}) with labels in the alphabet $X$.
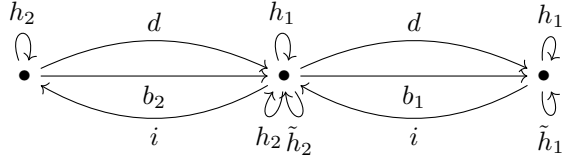
  \begin{center}
    \begin{tikzpicture}
      \begin{scope}
        \matrix (m) [matrix of math nodes, column sep=3cm, row sep=1cm]
                {\bullet & \bullet & \bullet\\
                  \mbox{}&\mbox{}&\mbox{}\\};
                \path[->] (m-1-1) edge [bend left=25] node [above] {$d$} (m-1-2);
                \path[->] (m-1-2) edge [bend left=25] node [above] {$d$} (m-1-3);
                \path[->] (m-1-2) edge node [below] {$b_1$} (m-1-3);
                \path[->] (m-1-1) edge node [below] {$b_2$} (m-1-2);
                \path[->] (m-1-2) edge [bend left=30] node [below] {$i$} (m-1-1);
                \path[->] (m-1-3) edge [bend left=30] node [below] {$i$} (m-1-2);
                \path[->] (m-1-2) edge [out=110,in=75,looseness=10] node [above] {$h_1$} (m-1-2);
                \path[->] (m-1-3) edge [out=95,in=65,looseness=10] node [above] {$h_1$} (m-1-3);
                \path[->] (m-1-1) edge [out=110,in=75,looseness=10] node [above] {$h_2$} (m-1-1);
                \path[->] (m-1-2) edge [out=265,in=235,looseness=10] node [below] {$h_2$} (m-1-2);
                \path[->] (m-1-3) edge [out=265,in=295,looseness=10] node [below] {$\tilde{h}_1$} (m-1-3);
                \path[->] (m-1-2) edge [out=305,in=275,looseness=10] node [below] {$\tilde{h}_2$} (m-1-2);
      \end{scope}
    \end{tikzpicture}
  \end{center}
  Either directly or by the package, we check that $f$ and each element of $F$ is compatible with the quiver.
  Denoting the vertices from left to right by $v_1,v_2,v_3$, we obtain the following signatures.
  \begin{gather*}
    \sigma(f_1)=\{(v_2,v_3)\}, \qquad
    \sigma(f_2)=\{(v_1,v_2)\},\\
    \sigma(f_3)=\{(v_3,v_3)\}, \qquad
    \sigma(f_4)=\{(v_2,v_2)\},\\
    \sigma(f_5)=\{(v_2,v_2),(v_3,v_3)\},\\
    \sigma(f)=\{(v_3,v_3)\}
  \end{gather*}
  To determine $\sigma(f)$, for example, notice that $\sigma(h_2i\tilde{h}_2h_1i\tilde{h}_1)=\{(v_3,v_1)\}$
  and that $\sigma(dd)=\sigma(b_1d)=\sigma(db_2)=\sigma(b_1b_2)=\{(v_1,v_3)\}$ and
  recall that $\sigma(1)$ contains all pairs of the form $(v_i,v_i)$.
\end{example}

\section{Q-consequences}
\label{sec:Q-csq}

The following definition characterizes the situations when a representation
of the claim in terms of the assumptions is also valid in terms of
operators. This generalizes the notion of $Q$-consequence given
in~\cite{Raab2019FormalPO}. Throughout the section, we fix a labelled quiver
$Q$ with labels in a set $X$.

\begin{definition}\label{def:Q-consequences}
  A $Q$-\emph{consequence} of $F\subseteq\RX$ is a polynomial $f\in\RX$, compatible with $Q$,
  such that there exist $g_i\in F$, $a_i,b_i\in\RX$, $1\leq i\leq n$,
  such that
  \begin{equation}\label{equ:Q-decompo}
    f=\sum_{i=1}^na_ig_ib_i,
  \end{equation}
  and for every $(u,v)\in\sigma(f)$ and every $i$, there exist vertices 
  $u_i,v_i$ such that $(u,u_i)\in\sigma(b_i)$, $(u_i,v_i)\in\sigma(g_i)$
  and $(v_i,v)\in\sigma(a_i)$.
\end{definition}

The conditions on the signatures mean that there exist three paths in
the quiver as illustrated in the following diagram.
\begin{center}
  \begin{tikzcd}
    u\ar[r, twoheadrightarrow, "f"]\ar[d, twoheadrightarrow, "b_i"' near start] & v\\
    u_i\ar[r, twoheadrightarrow, "g_i"'] & v_i\ar[u, twoheadrightarrow, "a_i"' near start]
  \end{tikzcd}
\end{center}

Proving that a given representation~\eqref{equ:Q-decompo} satisfies the
required conditions of the above definition is straightforward. In
Proposition~\ref{prop:Q-consequence_criterion}, we give an alternative
criterion for $Q$-consequences. This criterion will play an important
role later in Section~\ref{sec:general_rew_systems} on rewriting. Before,
we need the following lemma.

\begin{lemma}\label{lem:increasing_signatures}  
  Let $m\in\Span{X}$ be a monomial and $g\in\RX$ be a polynomial such
  that $\sigma(m)\subseteq\sigma(g)$.
  Then, for all monomials $a,b\in\Span{X}$, we have $\sigma(amb)\subseteq\sigma(agb)$.
  Moreover,
  for every $(u,v)\in\sigma(amb)$, there exist two vertices
  $\tilde{u},\tilde{v}$ such that $(u,\tilde{u})\in\sigma(b)$,
  $(\tilde{u},\tilde{v})\in\sigma(g)$, and $(\tilde{v},v)\in\sigma(a)$.
\end{lemma}

\begin{proof}
  For every $(u,v)\in\sigma(amb)$, there exists a path from $u$ to $v$
  with label $amb$. We split this path in 3 parts: the first part $\beta$
  has label $b$, the third part $\alpha$ has label $a$, and the second
  part has label $m$. Since $\sigma(m)\subseteq\sigma(g)$, for every
  $\tilde{m}\in\supp(g)$, there also exists a path $\gamma$ from
  $\tilde{u}:=t(\beta)$ to $\tilde{v}:=s(\alpha)$ with label $\tilde{m}$,
  as pictured on the following diagram
    \begin{center}
      \begin{tikzcd}
        u\ar[r, "b"] & \tilde{u} \ar[r, "m", bend left]
        \ar[r, "\tilde{m}"', bend right] & \tilde{v} \ar[r, "a"] & v
      \end{tikzcd}
    \end{center}
  Hence, $a\tilde{m}b$ is the label of $\alpha\gamma\beta$. Consequently,
  $\sigma(amb)\subseteq\sigma(a\tilde{m}b)$ for every
  $\tilde{m}\in\supp(g)$, and $(\tilde{u},\tilde{v})\in\sigma(g)$.
\end{proof}

\begin{proposition}\label{prop:Q-consequence_criterion}
  Let $F\subseteq\RX$ be a set of polynomials such that for every $g\in F$,
  there exists $m_g\in\supp(g)$ such that $\sigma(m_g)\subseteq\sigma(g)$.
  Let $f\in\RX$ be a compatible polynomial such that there exist
  $\lambda_i\in R$, $g_i\in F$, $a_i,b_i\in\Span{X}$, $1\leq i\leq n$,
  such that
  \begin{equation}\label{equ:Q-decompo_prop}
    f=\sum_{i=1}^n\lambda_ia_ig_ib_i,
  \end{equation}
  and for each $i$, we have $\sigma(f)\subseteq\sigma(a_im_{g_i}b_i)$.
  Then, $f$ is a $Q$-con\-se\-quence of $F$.
\end{proposition}

\begin{proof}
  By hypotheses, $f$ is compatible and for every $(u,v)\in\sigma(f)$ and
  for every $1\leq i\leq n$, we have $(u,v)\in\sigma(a_im_{g_i}b_i)$. Hence,
  using the hypothesis $\sigma(m_{g_i})\subseteq\sigma(g_i)$, from
  Lemma~\ref{lem:increasing_signatures}, there exist vertices $u_i$ and
  $v_i$ such that $(u,u_i)\in\sigma(b_i)$, $(u_i,v_i)\in\sigma(g_i)$ and
  $(v_i,v)\in\sigma(a_i)$. As a consequence, $f$ is a $Q$-consequence 
  of~$F$.
\end{proof}

Note that if for $m_g\in\supp(g)$, we have $\sigma(m_g)\subseteq\sigma(g)$,
then $\sigma(m_g)=\sigma(g)$ holds by definition.

\begin{example}\label{ex:Q-csq}
  Let us continue Example~\ref{ex:compatible}.
  We show that $f$ is a $Q$-con\-se\-quence of $F$ by considering the following representation:
  \begin{multline}\label{equ:Q-csq_decompo}
    f = f_1i\tilde{h}_1 + (d-b_1)f_2i\tilde{h}_2h_1i\tilde{h}_1 + f_3 +
    (d-b_1)f_4h_1i\tilde{h}_1\\
    +(d-b_1)h_2f_5\tilde{h}_2h_1i\tilde{h}_1 + h_1f_5\tilde{h}_1.
  \end{multline}
  Such a representation can be obtained with the package
  by tracking cofactors in polynomial reduction w.r.t.\ a monomial order.
  Here, we consider a degree-lexicographic order such that $d$ is greater
  than $h_i$'s and $b_i$'s. Then, $f$ can be reduced to zero using $F$,
  which gives~\eqref{equ:Q-csq_decompo}. Now, we have to check
  assumptions on signatures, either by checking
  Definition~\ref{def:Q-consequences} or the assumptions of
  Proposition~\ref{prop:Q-consequence_criterion}, both options are implemented in the package.
  For applying Proposition~\ref{prop:Q-consequence_criterion} by hand, we can choose $m_{f_1}=dh_1,m_{f_2}=h_2d,m_{f_3}=h_1\tilde{h}_1,m_{f_4}=h_2\tilde{h}_2$, and $m_{f_5}=di$, which satisfy $m_{f_i}\in\supp(f_i)$ and $\sigma(m_{f_i})=\sigma(f_i)$.
  Expanding~\eqref{equ:Q-csq_decompo} in the form~\eqref{equ:Q-decompo_prop}, we may check that $\sigma(a_im_{g_i}b_i)=\{(v_3,v_3)\}=\sigma(f)$ for every summand in the representation \eqref{equ:Q-decompo_prop}, which proves that $f$ is a $Q$-consequence of $F$.
\end{example}

To conclude this section, we prove that the property of being a $Q$-consequence is transitive, which we will exploit in Section~\ref{sec:compatible_GB}.

\begin{theorem}\label{thm:transitive}
  Let $F,G \subseteq \RX$ be sets of polynomials such that each element
  of $G$ is a $Q$-consequence of $F$. Then, any $Q$-consequence of $G$ is
  also a $Q$-consequence of $F$.
\end{theorem}

\begin{proof}
  Let $h$ be a $Q$-consequence of $G$, so that it is compatible with $Q$.
  Moreover, $h=\sum_ia_ig_ib_i$, with $g_i\in G$ and $a_i,b_i\in\RX$ such
  that for every $(u,v)\in\sigma(h)$ and every $i$, there exist
  vertices $u_i,v_i$ such that $(u,u_i)\in\sigma(b_i)$,
  $(u_i,v_i)\in\sigma(g_i)$ and $(v_i,v)\in\sigma(a_i)$. Since
  every element of $G$ is a $Q$-consequence of $F$, for each $g_i$, there
  exist $a_{i,j},b_{i,j}\in\RX$ and $f_{i,j}\in F$ such that
  $g_i=\sum_ja_{i,j}f_{i,j}b_{i,j}$ and for every
  $(u_i,v_i)\in\sigma(g_i)$ and every $j$, there exist
  $(u_{i,j},v_{i,j})\in\sigma(f_{i,j})$ such that
  $(u_i,u_{i,j})\in\sigma(b_{i,j})$ and
  $(v_{i,j},v_i)\in\sigma(a_{i,j})$. All together, we have
  \[h=\sum_i\sum_ja_ia_{i,j}f_{i,j}b_{i,j}b_i.\]
  For every $j$, $u_i$ and $v_i$ belong to $s(b_{i,j})\cap t(b_i)$ and
  $s(a_i)\cap t(a_{i,j})$, respectively, so that from Point 2 of
  Lemma~\ref{lem:signatures}, $(u,u_{i,j})\in\sigma(b_{i,j}b_i)$ and
  $(v_{i,j},v)\in\sigma(a_ia_{i,j})$, respectively. Hence, $h$ is a
  $Q$-consequence of~$F$.
\end{proof}

\section{Realizations}
\label{sec:realizations}

In this section, we formalize the translation of polynomials to
operators by substituting variables by basic operators. In particular, we
show in Theorem~\ref{thm:realizations} that being a $Q$-consequence is
enough to ensure that the corresponding operator identity can be inferred
from the assumed operator identities. To this end, we summarize the
relevant notions and basic facts from~\cite[Section 5]{Raab2019FormalPO}.

For a quiver $(V,E,s,t)$ and a ring $R$, $(\mathcal{M},\varphi)$ is
called a \emph{representation} of the quiver $(V,E,s,t)$, if
$\mathcal{M}=(\mathcal{M}_v)_{v\in V}$ is a family of $R$-modules and
$\varphi$ is a map that assigns to each $e\in E$ a $R$-linear map
$\varphi(e):\mathcal{M}_{s(e)}\to\mathcal{M}_{t(e)}$, see e.g.\
\cite{DerksenWeyman2005, GothenKing2005}. Not that any nonempty path $e_n{\dots}e_1$ in
the quiver induces a $R$-linear map
$\varphi(e_n){\cdot}{\dots}{\cdot}\varphi(e_1)$, since the maps
$\varphi(e_{i+1})$ and $\varphi(e_i)$ can be composed for every
$i \in \{1,\dots,n-1\}$ by definition of $\varphi$. Similarly, for every
$v \in V$, the empty path $\epsilon_v$ induces the identity map on
$\mathcal{M}_v$.

\begin{remark}\label{rem:categories}
 All notions and results of this section naturally generalize to $R$-linear
categories by considering objects and morphisms in such a category
instead of $R$-modules and $R$-linear maps, respectively.
 For more details, see Section~5.2 in \cite{Raab2019FormalPO}.
\end{remark}

\begin{definition}
  Let $R$ be a ring and let $Q$ be a labelled quiver with labelling $l$.
  We call a representation $(\mathcal{M},\varphi)$ of $Q$
  \emph{consistent} with the labelling $l$ if for any two nonempty paths
  $p=e_n{\dots}e_1$ and $q=d_n{\dots}d_1$ in $Q$ with the same source and
  target, equality of labels $l(p)=l(q)$ implies $\varphi(e_n){\cdot}
  {\dots}{\cdot}\varphi(e_1) = \varphi(d_n){\cdot}{\dots}{\cdot}
  \varphi(d_1)$ as $R$-linear maps.
\end{definition}

\begin{remark}
  If all paths with the same source and target have distinct labels, then
  every representation of that labelled quiver is consistent
  with its labelling. In particular, this holds if for every vertex all
  outgoing edges have distinct labels or analogously for incoming edges.
  These sufficient conditions can be verified without the need for 
  considering all possible paths.
\end{remark}

For Definition~\ref{def:realization} and
Lemma~\ref{lem:productrealizationsK}, we fix a ring $R$, a labelled
quiver $Q=(V,E,X,s,t,l)$ and a consistent representation
$\mathcal{R}=(\mathcal{M},\varphi)$ of $Q$. In order to define
realizations of a polynomial, we first need to introduce some notations.
Given two vertices $v,w$, we write $\RX_{v,w}$ for the set of polynomials
$f\in\RX$ such that $(v,w)\in\sigma(f)$. From Point 1 of
Lemma~\ref{lem:signatures}, $\RX_{v,w}$ is a module, and it is clear
that this module is free with basis the set of monomials $m$ such that
$(v,w)\in\sigma(m)$. We also denote by $\Hom_R(\mathcal{M}_v,\mathcal{M}_w)$
the set of $R$-linear maps from $\mathcal{M}_v$ to $\mathcal{M}_w$.

\begin{definition}\label{def:realization}
  For $v,w \in V$, we define the $R$-linear map
  $\varphi_{v,w}:\RX_{v,w}\to\Hom_R(\mathcal{M}_v,\mathcal{M}_w)$ by
 \[
 \varphi_{v,w}(l(e_n{\dots}e_1)) := \varphi(e_n)
        {\cdot}{\dots}{\cdot}\varphi(e_1)
 \]
 for all nonempty paths $e_n{\dots}e_1$ in $Q$ from $v$ to $w$ and, if
 $v=w$, also by $\varphi_{v,v}(1) := \id_{\mathcal{M}_v}$. For all
 $f\in\RX_{v,w}$, we call the $R$-linear map $\varphi_{v,w}(f)$ a
 \emph{realization} of $f$ w.r.t.\ the representation $\mathcal{R}$ of
 $Q$.
\end{definition}

Notice that the map $\varphi_{v,w}$ is well-defined since, by consistency
of $\mathcal{R}$, for every monomial $m\in\RX_{v,w}$, its realization $\varphi_{v,w}(m)$
does not depend on the path from $v$ to $w$ with label $m$.

In the proof of Theorem~\ref{thm:realizations}, we use an intermediate
result given in~\cite[Lemma 31]{Raab2019FormalPO}, whose statement is the
following.

\begin{lemma}\label{lem:productrealizationsK}
  Let $u,v,w \in V$. Then, for all $f\in\RX_{v,w}$ and $g\in\RX_{u,v}$, we
  have that $fg\in\RX_{u,w}$ and
 \[
  \varphi_{u,w}(fg) = \varphi_{v,w}(f){\cdot}\varphi_{u,v}(g).
 \]
\end{lemma}

\begin{theorem}\label{thm:realizations}
  Let $F\subseteq\RX$ be a set of polynomials and let $Q$ be a labelled quiver with labels in $X$.
  If a polynomial $f\in\RX$ is a $Q$-consequence of $F$, then for all consistent representations of the quiver $Q$ such that all realizations of all elements of $F$ are zero, all realizations
  of $f$ are zero.
\end{theorem}

\begin{proof}
  Assume that $f$ is a $Q$-consequence, so that it is compatible with $Q$
  and it can be written in the form $\sum a_ig_ib_i$, such that for each
  $(u,v)\in\sigma(f)$ and each $i$, there exist vertices $u_i,v_i$ such
  that $(u,u_i)\in\sigma(b_i)$, $(u_i,v_i)\in\sigma(g_i)$ and
  $(v_i,v)\in\sigma(a_i)$. Let us fix a consistent representation
  $\mathcal{R}=(\mathcal{M},\varphi)$ of $Q$. By linearity of
  $\varphi_{u,v}$ and from Lemma~\ref{lem:productrealizationsK}, we have
  \[\varphi_{u,v}(f)=\sum\varphi_{u,v}(a_ig_ib_i)=
  \sum\varphi_{v_i,v}(a_i){\cdot}\varphi_{u_i,v_i}(g_i){\cdot}
  \varphi_{u,u_i}(b_i).\]
  Hence, if all realizations of all elements of $F$ are zero, then 
  $\varphi_{u,v}(f)=0$, which means that all realizations of $f$ w.r.t
  $\mathcal{R}$ are zero.
\end{proof}

\begin{example}\label{ex:algebraic_proof_of_the_general_solution}
  We finish our proof of \eqref{ex:inhom_ODE} by considering certain representations of
  the quiver of Example~\ref{ex:compatible}. For a nonnegative integer $k$ and an open
  interval $I\subseteq\mathbb{R}$, we assign the spaces $C^k(I)$, $C^{k+1}(I)$,
  and $C^{k+2}(I)$ to the the vertices from right to left. Hence, differentiation
  and integration induce operators $\partial:C^{k+1}(I)\to C^k(I),
  \partial:C^{k+2}(I)\to C^{k+1}(I),\int_1:C^k(I)\to C^{k+1}(I)$, and
  $\int_2:C^{k+1}(I)\to C^{k+2}(I)$. We also assume the following regularity
  of functions: $B_1$ is $C^k$, $H_1$ and $B_2$ are
  $C^{k+1}$ and $H_2$ is $C^{k+2}$ on $I$. Then, the natural representation associated
  with these operators is consistent. Moreover, we have seen in Example~\ref{ex:Q-csq}
  that $f$ is a $Q$-consequence of $F$. Since all realizations of
  $f_i$'s are zero, by Theorem~\ref{thm:realizations}, all realizations of $f$
  are zero. In particular, for every nonnegative integer $k$ and every
  $r(x)\in C^k(I)$, the function $y(x)$ defined by~\eqref{equ:sol} is a solution of the
  inhomogeneous differential equation~\eqref{ex:inhom_ODE}.

  Instead of considering scalar differential equations we could consider differential
  systems of the form~\eqref{ex:inhom_ODE} for vector-valued functions $y(x)$ of
  arbitrary dimension $n$. More explicitly, we can also consider coefficients
  $B_1(x),B_2(x)$ as $n\times n$ matrices, $r(x)$ as a vector of dimension $n$, and
  $H_1(x)$ and $H_2(x)$ as fundamental matrix solutions of the homogeneous systems
  $y^\prime(x)-B_i(x)y(x)=0$. We still obtain consistent representations of the quiver
  where the vertices are mapped to $C^k(I)^n, C^{k+1}(I)^n$ and $C^{k+2}(I)^n$,
  respectively. Then, Theorem~\ref{thm:realizations} immediately proves that
  the function $y(x)$ defined by~\eqref{equ:sol} is a solution of the
  inhomogeneous differential equation~\eqref{ex:inhom_ODE}.
  Similarly, analogous statements for other suitable functional spaces can be proven just by choosing different representations of the quiver.
\end{example}

\section{Compatible rewriting}
\label{sec:general_rew_systems}

In this section, we give conditions on polynomials such that rewriting 
to zero of a compatible polynomial by them proves that it is a $Q$-consequence. First,
we recall from~\cite[Definition 2]{Raab2019FormalPO} a general notion of rewriting
one polynomial by another in terms of an arbitrary monomial division.
Notice that the standard polynomial reduction is a particular case, where
$m$ is the leading monomial of $g$ w.r.t.\ a monomial order and $\lambda$ is such
that $amb$ is cancelled in~\eqref{equ:rew_candidate}. 

\begin{definition}
  Let $g\in\RX$ be a polynomial and let $m\in\supp(g)$. Let $f\in\RX$ be
  a polynomial such that $m$ divides some monomial $m_f\in\supp(f)$,
  i.e., $m_f=amb$ for monomials $a,b\in\Span{X}$. For every
  $\lambda\in R$, we say that $f$ {\em can be rewritten} to
  \begin{equation}\label{equ:rew_candidate}
    h:=f+\lambda agb,
  \end{equation}
  using $(g,m)$.
\end{definition}

We fix a labelled quiver $Q$ with labels in $X$. It turns out that to obtain
$Q$-consequences using rewriting (Theorem~\ref{thm:consequences}), we  
need to choose suitable divisor monomials such that signatures only
increase. In particular, this is the case when divisor monomials have
minimal signature, as stated in the following lemma.

\begin{lemma}\label{lem:rew_and_signatures}
  Let $g\in\RX$ be a polynomial and let $m\in\supp(g)$
  be such that $\sigma(m)=\sigma(g)$. If $f$ can be rewritten to 
  $h=f+\lambda agb$ using $(g,m)$, then
  \[\sigma(f)\subseteq\sigma(h)\quad \text{and}\quad
  \sigma(f)\subseteq\sigma(amb).\]
\end{lemma}

\begin{proof}
  By definition of signatures, $\sigma(f)\subseteq\sigma(amb)$. By
  Lemma~\ref{lem:increasing_signatures} and from $\sigma(m)=\sigma(g)$,
  we have $\sigma(amb)\subseteq\sigma(agb)$. Altogether, $\sigma(f)$ is
  included in $\sigma(agb)$, which itself is contained in
  $\sigma(\lambda agb)$. From 1. of Lemma~\ref{lem:signatures}, we deduce
  $\sigma(f)\subseteq\sigma(h)$.
\end{proof}

Now, we define the rewriting relation induced by a fixed choice of
divisor monomials and its compatibility with a quiver.
For any rewriting relation we denote single rewriting steps by $\to$ and the reflexive transitive closure by $\rewTrans$.

\begin{definition}\label{def:compatible_rew_sys}
  Let $G\subseteq\RX$ be a set of polynomials and let $\DM:G\to\PX$ be a
  function from $G$ to the power set of $\Span{X}$, such that
  $\DM(g)\subseteq\supp(g)$, for every $g\in G$.
  \begin{enumerate}
  \item For $g\in G$, we say that $m\in\DM(g)$ is a {\em divisor
    monomial} of $g$ w.r.t.\ $\DM$.
  \item We say that $f$ {\em rewrites} to $h$ by $(G,\DM)$, denoted as
    $f\to_{G,\DM}h$, if there exists $g\in G$ and a divisor monomial
    $m\in\DM(g)$ such that $f$ can be rewritten to $h$ using $(g,m)$.
  \item We say that $\DM$ is {\em compatible} with a labelled quiver Q if
    for every $g\in G$ and every $m\in\DM(g)$, we have
    $\sigma(m)=\sigma(g)$.
  \end{enumerate}
\end{definition}

From now on, we fix a set of polynomials $G\subseteq\RX$ as well as a map
$\DM$ selecting divisor monomials.

\begin{remark}\label{rem:DM}
  Notice that there exist two extreme cases for the definition of $\DM$:
  \begin{enumerate}
  \item $\DM$ selects exactly one monomial for each $g\in G$, for instance,
    the leading monomial $\LM(g)$ w.r.t.\ a monomial order, see the example
    in Section~\ref{sec:compatible_GB}.
  \item All monomials in $\supp(g)$ are divisor monomials.
    Then, $\to_{G,\DM}$ coincides with the rewriting relation introduced
    in~\cite[Definition 2]{Raab2019FormalPO}, for which ideal membership is
    equivalent to reduction to zero~\cite[Lemma 4]{Raab2019FormalPO}.
    Moreover, if such a $\DM$ is compatible with $Q$, then all polynomials in
    $G$ are uniformly compatible, i.e., every monomial of a polynomial has the same signature.
    The following theorem gives a generalization of Corollary~17 in \cite{Raab2019FormalPO}.
  \end{enumerate}
\end{remark}

\begin{theorem}\label{thm:consequences}
  Let $G\subseteq\RX$ be a set of polynomials and let $\DM$ be a function
  selecting divisor monomials as in
  Definition~\ref{def:compatible_rew_sys}. Let $f\in\RX$ be a polynomial
  such that $f\rewTrans_{G,\DM}0$. Then, for every labelled quiver $Q$
  with labels $X$ such that $\DM$ is compatible with $Q$, we have that
  \[f\text{ is compatible with }Q \iff f\text{ is a $Q$-consequence of }G.\]
\end{theorem}

\begin{proof}
  Since $f$ rewrites to zero, there exists a sequence
  $f=h_0\to h_1\to\cdots\to h_n=0$. Hence, there exist $\lambda_i\in R$,
  $a_i,b_i\in\Span{X}$, $g_i\in G$, and $m_i\in\DM(g_i)$ such that
  $h_i=h_{i-1}+\lambda_ia_ig_ib_i$ and
  $a_im_ib_i\in\supp(h_{i-1})$. Hence, $f$ can be written as
  $f=\sum_{i=1}^n-\lambda_ia_ig_ib_i$. From
  Lemma~\ref{lem:rew_and_signatures}, we conclude inductively that
  $\sigma(f)\subseteq\sigma(h_{i-1})\subseteq\sigma(a_im_ib_i)$. Hence,
  if $f$ is compatible with $Q$, then $f$ is a $Q$-con\-se\-quence of $G$ by
  Proposition~\ref{prop:Q-consequence_criterion}.  Conversely, if $f$ is
  a $Q$-consequence of $G$, then it is compatible by definition.
\end{proof}

\begin{example}\label{ex:non_compatible_DM}
  Let us translate Example~\ref{ex:Q-csq} in the language introduced in
  this section. The leading monomials w.r.t.\ the degree-lexicographic order
  used in that example can be understood as the divisor monomials selected by
  the function $\DM$ defined on $F$ such that
  $\DM(f_i)=\{\LM(f_i)\}$ holds for all $i$. In particular,
  \begin{gather*}
    \DM(f_1)=\{dh_1\},\quad \DM(f_2)=\{dh_2\},\\
    \DM(f_3)=\{h_1\tilde{h}_1\},\quad\DM(f_4)=\{h_2\tilde{h}_2\},
    \quad\DM(f_5)=\{di\}.
  \end{gather*}
  Then, $\DM$ is not compatible with $Q$, since $\sigma(f_2)=\{(v_1,v_2)\}$
  is not equal to $\sigma(dh_2)=\{(v_1,v_2),(v_2,v_3)\}$. Hence, we cannot 
  apply Theorem~\ref{thm:consequences} to show that $f$ is a $Q$-consequence
  of $F$ even though $f\overset{*}{\to}_{F,\DM}0$. So, we need to look at the
  explicit representation of $f$ induced by this reduction, which was already done in
  Example~\ref{ex:Q-csq}. In order to apply Theorem~\ref{thm:consequences},
  we need to redefine $\DM$ so that it is compatible with $Q$. In
  particular, we need to impose $\DM(f_2)\subseteq\{h_2d,b_2h_2\}$. If
  $b_2h_2\in\DM(f_2)$, then $f\overset{*}{\to}_{F,\DM}0$, which gives another
  proof that $f$ is a $Q$-consequence of $F$ based on
  Theorem~\ref{thm:consequences}. Otherwise, if $\DM(f_2)=\{h_2d\}$, then
  $f$ is irreducible w.r.t.\ $\overset{*}{\to}_{F,\DM}$. Therefore, we need
  to complete $F$ with $Q$-consequences of it such that $\DM$ remains
  compatible with $Q$ and $f$ reduces to zero, which is the topic of the next
  section.
\end{example}

\section{Compatible reductions and partial \Gr\ bases}
\label{sec:compatible_GB}

In this section, we discuss standard noncommutative polynomial reduction
as a special case of the rewriting approach from the previous section.
Since in the noncommutative case, \Gr\ bases are not necessarily finite,
see~\cite{MR1299371}, we also have to work with partial \Gr\ bases which are
obtained by finitely many iterations of the Buchberger procedure.
We adapt the noncommutative Buchberger procedure for computing (partial) \Gr\ bases that can be used for compatible rewriting.

In what follows, $R$ is assumed to be a field $\K$ and we fix a 
monomial order $\leq$ on $\Span{X}$, that is, a well-founded total order
compatible with multiplication on $\Span{X}$. We also fix a labelled quiver
$Q$ with labels in $X$ and a set of polynomials $F\subseteq\KX$. Given a set
of polynomials $G\subseteq\KX$, one step of the standard polynomial reduction
w.r.t.\ $G$ is denoted by $f\to_Gh$.

As explained in Remark~\ref{rem:DM}, the monomial order induces the $\DM$
function that selects leading monomials of a set $G\subseteq\KX$. This $\DM$
function is compatible with $Q$ if and only if all elements of $G$ are $Q$-order
compatible in the following sense.

\begin{definition}
  A compatible polynomial $f$ is said to be $Q$-\emph{order compatible}
  if $\sigma(\LM(f))=\sigma(f)$. 
\end{definition}

By transitivity of $Q$-consequences, see
Theorem~\ref{thm:transitive}, and Theorem~\ref{thm:consequences}, we
obtain the following statement.

\begin{corollary}\label{cor:reduction_crit}
  Let $F\subseteq\KX$, $G\subseteq(F)$, and $f\in\KX$ such that $f\rewTrans_G0$. Then,
  for all labelled quivers $Q$ such that all elements of $G$ are both
  $Q$-consequences of $F$ and $Q$-order compatible, we have
  \[f\text{ is compatible with }Q \iff f\text{ is a $Q$-consequence of }F.\]
\end{corollary}

\begin{remark}\label{lem:Qord_is_mon_ord}
  For polynomials, being $Q$-order compatible can also be interpreted in
  terms of a partial monomial order. Given $m,m'\in\Span{X}$, we define
  $m\Qord m'$ if $m\leq m'$ and $\sigma(m')\subseteq\sigma(m)$.
  The partial order $\Qord$ respects multiplication of monomials since, by Lemma~\ref{lem:increasing_signatures}, $\sigma(m')\subseteq\sigma(m)$ implies $\sigma(am'b)\subseteq\sigma(amb)$ for all
  $a,b\in\Span{X}$. Then, $f$ is $Q$-order compatible if and only if
  $\supp(f)$ admits a greatest element for $\Qord$.
\end{remark}

Candidates for $G$ as in Corollary~\ref{cor:reduction_crit} are
partial \Gr\ bases that are computed by the noncommutative Buchberger
procedure~\cite{buchberger1965algorithmus, MR1299371}. However, in view of
the assumptions, we only add reduced $S$-polynomials that are both
$Q$-consequences of $F$ and $Q$-order compatible in each iteration.
Checking $Q$-order compatibility is easy. Selecting $Q$-consequences is
harder since we do not want to use explicit representations as in 
Definition~\ref{def:Q-consequences}. Instead, we propose a simpler
criterion based on the following lemma and discussion.

First, we recall some terminology and fix notations for $S$-poly\-no\-mials.
Let $G\subseteq\KX$. \emph{Ambiguities} of $G$ defined in~\cite{Bergman1978},
also called compositions in~\cite{MR0506423}, 
are given by minimal overlaps or inclusions of the two leading monomials $\LM(g)$ and
$\LM(g')$, where $g$ and $g'$ belong to $G$. Formally, each ambiguity can be
described by a $6$-tuple $\amb=~(g,g',a,b,a',b')$, where $a,b,a',b'$
are monomials such that, among other conditions, we have
\[a\LM(g)b=a'\LM(g')b'.\]
This monomial is called the {\em source} of $\amb$
and the $S$-polynomial of $\amb$ is $\SP(\amb):=agb-a'g'b'$,
cf.~\cite{MR1299371}.

\begin{lemma}\label{lem:comp_SP}
  Let $G\subseteq\KX$ be a set of $Q$-order compatible polynomials and
  let $s$ be a $S$-polynomial of $G$ with source a compatible monomial
  $m\in\Span{X}$. Then $\sigma(m)\subseteq\sigma(s)$. If moreover,
  $s\rewTrans_G\hat{s}$ with $\sigma(\hat{s})\subseteq\sigma(m)$, then
  $\sigma(s)=\sigma(\hat{s})=\sigma(m)$ and $\hat{s}$ is a
  $Q$-consequence of $G$.
\end{lemma}

\begin{proof}
  Since $s$ is a $S$-polynomial of $G$ of source $m$, there exist
  $g,g'\in G$ and monomials $a,a',b,b'\in\Span{X}$ such that
  $s=(m-agb)-(m-a'g'b')$ with $a\LM(g)b=a'\LM(g')b'=m$.
  
  Let us prove the first assertion.  The polynomials $g$ and $g'$ being $Q$-order
  compatible, we have $\sigma(\LM(g))=\sigma(g)$ and
  $\sigma(\LM(g'))=\sigma(g')$. Hence,
  from Lemma~\ref{lem:increasing_signatures}, we have
  \begin{equation}\label{equ:sigma_SP}
    \begin{gathered}
      \sigma(m)=\sigma(a\LM(g)b)\subseteq\sigma(agb),\\
      \sigma(m)=\sigma(a'\LM(g')b')\subseteq\sigma(a'g'b').
    \end{gathered}
  \end{equation}
  From this and $s=a'g'b'-agb$, we get that
  $\sigma(m)\subseteq\sigma(s)$.

  Now, we assume that $s\rewTrans_G\hat{s}$, that is there is a rewriting
  sequence
  \[s_0=s\to_Gs_1\to_{G}\cdots\to_Gs_n=\hat{s},\]
  so that there exist $g_i\in G$, $a_i,b_i\in\Span{X}$ and
  $\lambda_i\in\K$, $1\leq i\leq n$, such that
  $a_i\LM(g_i)b_i\in\supp(s_i)$ and $s_{i+1}=s_i+\lambda_ia_ig_ib_i$, so
  that we have $\hat{s}-s=\sum_i\lambda_ia_ig_ib_i$ and
  \begin{equation}\label{equ:decompo_s_hat}
    \hat{s}=\sum_{i=1}^{n}\lambda_ia_ig_ib_i+a'g'b'-agb.
  \end{equation}
  Using inductively $s_i\to_G s_{i+1}$ and
  Lemma~\ref{lem:rew_and_signatures}, we get
  \begin{equation}\label{equ:sigma_s}
    \sigma(s)\subseteq\sigma(\hat{s})\quad\text{and}\quad
    \sigma(s)\subseteq\sigma(a_i\LM(g_i)b_i),
  \end{equation}
  so that we have
  \[\sigma(m)\subseteq\sigma(s)\subseteq\sigma(\hat{s}).\]
  If, moreover $\sigma(\hat{s})\subseteq\sigma(m)$, then we get the
  following sequence of inclusions:
  \[\sigma(\hat{s})\subseteq\sigma(m)\subseteq\sigma(s)\subseteq
  \sigma(\hat{s}).\]
  Hence, the equality $\sigma(\hat{s})=\sigma(s)=\sigma(m)$ holds. Now,
  we show that $\hat{s}$ is a $Q$-consequence using 
  Proposition~\ref{prop:Q-consequence_criterion}. Since the elements of
  $G$ are $Q$-order compatible, we have
  $\sigma(\LM(\tilde{g}))=\sigma(\tilde{g})$, for all $\tilde{g}\in G$.
  Moreover, since $m$ is compatible and $\sigma(\hat{s})=\sigma(m)$, 
  $\hat{s}$ is compatible. Finally, from~\eqref{equ:sigma_SP} and
  \eqref{equ:sigma_s}, we have the following inclusions:
  $\sigma(\hat{s})\subseteq\sigma(a\LM(g)b)$,
  $\sigma(\hat{s})\subseteq\sigma(a'\LM(g')b')$ and
  $\sigma(\hat{s})\subseteq\sigma(a_i\LM(g_i)b_i)$. As a conclusion,
  $\hat{s}$ is a $Q$-consequence of $G$.
\end{proof}

Starting with a set of $Q$-order compatible polynomials $F$, we apply
this lemma in the case where $G$ is the partial \Gr\ basis
computed in the current iteration of the completion procedure. In
particular, if a reduced $S$-polynomial $\hat{s}$ satisfies 
$\sigma(\hat{s})\subseteq\sigma(m)$ as in the lemma, then it is a
$Q$-consequence of $G$. By transitivity, it is then also a
$Q$-consequence of $F$, which follows from the following observation.

\begin{remark}\label{rem:inductive_Q-csq}
  Consider a set $F\subseteq\KX$ of compatible polynomials and a family
  of sets $G_i$ inductively defined by $G_0=\emptyset$ and
  $G_{i+1}=G_i\cup\{g_{i+1}\}$, where $g_{i+1}$ is a $Q$-consequence of
  $F\cup G_i$. Using inductively transitivity of $Q$-consequences proven in Theorem~\ref{thm:transitive}, we
  obtain that, for each $i$, all elements of $F\cup G_i$ are
  $Q$-consequences of $F$.
\end{remark}

In summary, we obtain the following adaptation of the noncommutative version
of Buchberger's procedure for computing a partial \Gr\ basis composed
of elements that are both $Q$-con\-se\-quen\-ces of $F$ and $Q$-order
compatible. At each step, we select an $S$-polynomial $s$ whose source
$m$ is a compatible monomial, and we keep a reduced form $\hat{s}$ only
if it is $Q$-order compatible and $\sigma(\hat{s})\subseteq\sigma(m)$.
This procedure is implemented in the \textsc{Mathematica} package \texttt{OperatorGB}.
Note that
since the Buchberger procedure does not terminate in general for noncommutative polynomials, also our
adaptation of it is not guaranteed to terminate.

Notice that the completion procedure described above can be slightly generalized by not
necessarily computing reduced forms of $S$-polynomials. Instead, we only
reduce an $S$-polynomial as long as it remains a $Q$-consequence, see
the discussion above, and it remains $Q$-order compatible.
This is stated formally in Procedure~\ref{proc:Q-completion}.

\begin{algorithm}
  \caption{$Q$-order compatible completion}
  \label{proc:Q-completion}
  \begin{algorithmic}[1]
    \REQUIRE $F\subseteq\KX$, a labelled quiver $Q$ with labels in $X$, and a monomial order $\leq$ such that every $f\in F$ is $Q$-order compatible
    \ENSURE $G\supseteq F$ a set of $Q$-consequences of $F$ that are $Q$-order compatible
    \STATE $P:=$ ambiguities of $F$; $G:=F$
    \WHILE{$P\neq\emptyset$}
    \STATE choose $\amb\in P$
    \STATE $P:=P\setminus\{\amb\}$; $s:=\SP(\amb)$;
    $m:=$ the source of $\amb$
    \IF{$m$ is compatible \AND $\sigma(s)\subseteq\sigma(m)$ \AND $s$ is $Q$-order compatible}
    \WHILE{$\exists s': s\to_Gs'$}
    \IF{$s'=0$}
    \STATE \textbf{go to} 2 (i.e., break the outer {\bf if} statement)
    \ELSIF{$\sigma(s')\subseteq\sigma(m)$ \AND $s'$ is $Q$-order compatible}
    \STATE $s:=s'$
    \ELSE
    \STATE \textbf{go to} 15 (i.e., break the inner {\bf while} loop)
    \ENDIF
    \ENDWHILE
    \STATE $G:=G\cup\{s\}$
    \STATE $P:=P\cup\{\text{ambiguities created by}\ s\}$
    \ENDIF
    \ENDWHILE
    \RETURN{$G$}
  \end{algorithmic}
\end{algorithm}

Due to the checks in line $9$, each element $g$ of the output $G$ of the
procedure is both a $Q$-consequence of $F$ and $Q$-order compatible. In
summary, we have shown that our procedure is correct.

\begin{theorem}
  Let $F\subseteq\KX$ be a set of polynomials, let $Q$ be a labelled
  quiver with labels in $X$ and let $\leq$ be a monomial order such that
  each element of $F$ is $Q$-order compatible.
  Then, each element of the output $G$ of Procedure~\ref{proc:Q-completion} is both a $Q$-consequence of $F$ and $Q$-order
  compatible.
\end{theorem}

\begin{example}
  Let us continue Example~\ref{ex:non_compatible_DM} in the case
  $\DM(f_2)=\{h_2d\}$. For that, we consider the field $\K=\mathbb{Q}$ and a degree-lexicographic order
  such that $d<b_2<h_2$ and $d$ is greater than $b_1$ and $h_1$. Then,
  choosing the ambiguity $(f_2,f_5,1,i,h_2,1)$, the first iteration of the
  outer loop in Procedure~\ref{proc:Q-completion} yields
  $G:=F\cup\{b_2h_2i-dh_2i+h_2\}$. With this $G$, we have
  $f\overset{*}{\to}_G0$. From this reduction to $0$, and since $f$ is
  compatible with $Q$, $f$ is a $Q$-consequence of $F$ by
  Corollary~\ref{cor:reduction_crit}. These computations can also be done
  by the package.
\end{example}

\begin{remark}
\label{rem:uniquelabels}
 We consider the special case when all edges of $Q$ have unique labels.
 Then, all non-constant monomials have at most one element in their signature.
 Therefore, every compatible polynomial is $Q$-order compatible for any monomial order, since the monomial $1$ is the smallest.
 Moreover, one can show easily that the source of an ambiguity of two polynomials is compatible whenever these two polynomials are compatible with $Q$.
 In addition, from Lemma~\ref{lem:rew_and_signatures}, it follows that polynomial reduction of compatible polynomials by compatible ones does not change the signature unless the result of the reduction lies in $\K$.
 Altogether, Procedure~\ref{proc:Q-completion} reduces to the standard Buchberger procedure (i.e., without checking signatures and compatibility during computation) as long as no $S$-polynomial is (or is reduced to) a nonzero constant.
 In other words, we have the following theorem, which, together with Theorem~\ref{thm:realizations}, gives a generalization of Theorem~1 in \cite{Raab2019FormalPO}.
\end{remark}

\begin{theorem}\label{thm:uniquelabels}
 Assume that edges of $Q$ have unique labels.
 Let $F\subseteq\KX$ be a set of compatible polynomials and let $G$ be a (partial) \Gr\ basis computed by the standard Buchberger procedure (i.e., disregarding $Q$ during computation).
 If $G$ does not contain a constant polynomial, then for every polynomial $f \in \KX$ such that $f \rewTrans_G 0$, we have
 \[f\text{ is compatible with }Q \iff f\text{ is a $Q$-consequence of }F.\]
 Moreover, if $1\not\in(F)$, then this equivalence holds for every $f\in(F)$.
\end{theorem}

\section{Summary and discussion}

By Theorem~\ref{thm:realizations}, for proving new operator identities from known ones, it suffices to show that the corresponding polynomials are $Q$-consequences.
In practice, there are several options to prove that a compatible polynomial $f$ is a $Q$-consequence of some set $F$ of compatible polynomials.
Each of these options can be turned into a certificate that $f$ is a $Q$-consequence of $F$.
Given an explicit representation of $f$ in terms of $F$ of the form \eqref{equ:Q-decompo}, one can either check Definition~\ref{def:Q-consequences} directly, or expand cofactors into monomials and apply Proposition~\ref{prop:Q-consequence_criterion}.
Alternatively, using compatible rewriting, if $f \rewTrans_{F,\DM} 0$ and the selection of divisor monomials by $\DM$ is compatible with $Q$, then $f$ is a $Q$-consequence by Theorem~\ref{thm:consequences}.
Altogether, from Theorems~\ref{thm:realizations} and~\ref{thm:consequences}, we
immediately obtain the following.

\begin{corollary}\label{cor:prove_identities}
  Let $F$ be a set of polynomials, $\DM$ a function selecting divisor monomials, and $f$ a polynomial such that $f\rewTrans_{F,\DM}0$.
  Then, for all labelled quivers $Q$ such that $f$, $F$, and $\DM$ are compatible with $Q$ and for all consistent representations of $Q$ such that all realizations of all elements of $F$ are zero, all realizations of $f$ are zero.
\end{corollary}
Note that rewriting to zero w.r.t.\ $F$ and $\DM$ is independent of the quiver $Q$. In particular, if the above corollary is interpreted in terms of $R$-linear categories, the main result of \cite{Raab2019FormalPO}, Theorem~32, is obtained as a special case by Remark~\ref{rem:DM}.

More generally, if one cannot verify that $f$ can be rewritten to zero by $F$, there still might exist a set $G$ of $Q$-consequences of $F$ with divisor monomials selected by some $\DM$ such that $f\rewTrans_{G,\DM}0$ and Theorem~\ref{thm:transitive} can be applied.
Algorithmically, based on suitable monomial orderings, Procedure~\ref{proc:Q-completion} produces candidates $G$ such that Corollary~\ref{cor:reduction_crit} can be used to prove that $f$ is a $Q$-consequence of $F$ by standard polynomial reduction.

Notice that Procedure~\ref{proc:Q-completion} can be extended in various directions.
For example, in order to systematically generate more $Q$-consequences, reduced $S$-polynomials that are not $Q$-order compatible could be collected in a separate set, which should not be used for constructing and reducing new $S$-polynomials.
Instead of fixing a monomial ordering from the beginning, one might start with a partial ordering that is then extended during the completion procedure in order to make obtained $S$-polynomials $Q$-order compatible.
More generally, without any partial ordering on monomials, one might even consider compatible functions $\DM$ which not necessarily select only one divisor monomial per polynomial and aim at completing the induced rewriting relation. However, termination of such rewriting relations is an issue.
Finally, another topic for future research is to generalize the results of this paper to tensor reduction systems used for modelling linear operators as described in \cite{HosseinPoorRaabRegensburger2018}.

\bibliographystyle{plain}
\bibliography{Quivers_GB} 

\end{document}